\documentclass[a4paper,UKenglish,cleveref, autoref, thm-restate]{lipics-v2021}

\hideLIPIcs  

\graphicspath{{./figures/}}

\title{Flip Graphs of Pseudo-Triangulations With Face Degree at Most 4} 

\author{Maarten L\"offler}{Utrecht University, the Netherlands}{m.loffler@uu.nl}{}{}

\author{Tamara Mchedlidze}{Utrecht University, the Netherlands}{t.mtsentlintze@uu.nl}{}{}
\author{David Orden}{University of Alcal\'a, Spain}{david.orden@uah.es}{}{Supported by Project PID2019-104129GB-I00 funded by MCIN/ AEI /10.13039/501100011033.}
\author{Josef Tkadlec}{Charles University, Czech Republic}{josef.tkadlec@iuuk.mff.cuni.cz}{}{Supported by the Center for Foundations of Modern Computer Science (Charles University project UNCE/SCI/004) and by project PRIMUS/24/SCI/012 from Charles University.}
\author{Jules Wulms}{TU Eindhoven, the Netherlands}{j.j.h.m.wulms@tue.nl}{}{}

\authorrunning{M. L\"offler, T. Mchedlidze, D. Orden, J. Tkadlec, and J. Wulms} 

\Copyright{Maarten L\"offler, Tamara Mchedlidze, David Orden, Josef Tkadlec, and Jules Wulms} 



\acknowledgements{This work was initiated during the GG Week 2023, held at the University of Alcal\'a, Spain.}

\nolinenumbers 

\newcommand{\reals}{\mathbb{R}}

\begin{document}

\maketitle

\begin{abstract}
A \emph{pseudo-triangle} is a simple polygon with exactly three convex vertices, and all other vertices (if any) are distributed on three concave chains.
A \emph{pseudo-triangulation}~$\mathcal{T}$ of a point set~$P$ in~$\reals^2$ is a partitioning of the convex hull of~$P$ into pseudo-triangles, such that the union of the vertices of the pseudo-triangles is exactly~$P$.
We call a size-4 pseudo-triangle a \emph{dart}.
For a fixed $k\geq1$, we study $k$-dart pseudo-triangulations ($k$-DPTs), that is, pseudo-triangulations in which exactly $k$ faces are darts and all other faces are triangles.
We study the flip graph for such pseudo-triangulations, in which a flip exchanges the diagonals of a pseudo-quadrilatral.
Our results are as follows.
We prove that the flip graph of $1$-DPTs is generally not connected, and show how to compute its connected components. 
Furthermore, for $k$-DPTs on a point configuration called the \emph{double chain} we analyze the structure of the flip graph on a more fine-grained level. 
\end{abstract}

\section{Introduction}

A \emph{pseudo-triangle} is a simple polygon with exactly three convex vertices, and (if any) all other vertices distributed on three concave chains. A \emph{pseudo-triangulation}~$\mathcal{T}$ of a point set~$P$ in~$\reals^2$ is a partitioning of the convex hull of~$P$ into pseudo-triangles, such that the union of the vertices of the pseudo-triangles is exactly~$P$. 

Pseudo-triangulations were introduced in the early 1990's by Pocchiola and Vegter to study the visibility complex of disjoint convex regions~\cite{pocchiola1996visibility} and by Chazelle et al.\ for ray shooting in polygons~\cite{chazelle1994ray}. It was in the early 2000's that pseudo-triangulations of point sets became popular, when Streinu showed that \emph{pointed} pseudo-triangulations of point sets, those in which every vertex is \emph{pointed}, i.e., incident to an angle larger than~$\pi$, are minimally rigid~\cite{streinu2005pseudo} and used this for a solution of the Carpenter's Rule Problem. The converse statement that every planar minimally rigid graph admits a drawing as a pointed pseudo-triangulation was proved by Haas et al.~\cite{haas2003planar}, later generalized to non-minimally rigid and non-pointed pseudo-triangulations by Orden et al.~\cite{orden2007combinatorial} using the notion of \emph{combinatorial pseudo-triangulation}, an embedding of a planar graph together with a labelling of the angles mimicking the properties of angles in a geometric pseudo-triangulation.

Among the many other results on pseudo-triangulations, for which we refer to the survey by Rote et al.~\cite{rote2008pseudo}, let us highlight the notion of a \emph{flip}~\cite{aichholzer2003flips,orden2005polytope}. There are three types of flips. The first one  follows the spirit of flips in triangulations, exchanging the only interior edge in the two geodesic diagonals of a pseudo-quadrilateral, as in Figure~\ref{fig:flips-example}a-b.
This includes the case of a degenerate pseudo-quadrilateral,  with two consecutive corners merged into a single one, as in Figure~\ref{fig:flips-example}c where the two lowest vertices from Figure~\ref{fig:flips-example}b have been merged. The two remaining types of flip insert or remove an interior edge to obtain another pseudo-triangulation, respectively increasing or decreasing by one the number of pointed vertices (this will be fixed in the present work, so such types of flip will not appear). The flip graph for pseudo-triangulations turns out to be connected and have diameter in~$O(n\log n)$~\cite{bereg2004transforming}. 

\begin{figure}
    \centering
    \includegraphics{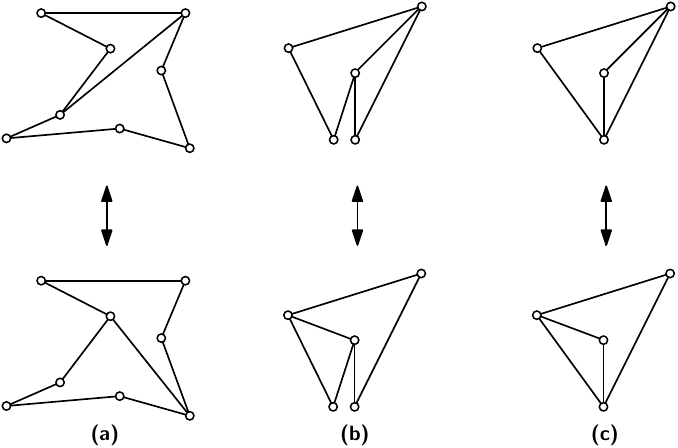}
    \caption{Some flips in pseudo-triangulations.}
    \label{fig:flips-example}
\end{figure}

The fact that a pseudo-triangle can have linear size and, therefore, the flip operation cannot be computed in constant time as for triangulations, led to the consideration of pseudo-triangulations with bounded size of the internal faces. In particular, Kettner et al.~\cite{kettner2003tight} showed that for every point set there is a pointed pseudo-triangulation with internal faces of size~3 or~4, called a \emph{4-pointed pseudo-triangulation} or \emph{4-PPT}. These 4-PPTs fulfill nice properties, like being properly 3-colorable while that question is NP-complete for general pseudo-triangulations~\cite{aichholzer20153colorability}. By contrast, some properties known for general pseudo-triangulations turned out to be elusive for 4-PPTs. In particular, a long-standing open problem is whether the flip graph of 4-PPTs is connected, which has only been proved for combinatorial 4-PPTs~\cite{aichholzer2014flips}. The aim of this work is to generalize this problem and prove results on cases that can provide additional insight towards solving that open problem.

We consider flips in \emph{4-pseudo-triangulations} or \emph{4-PT}s, which are defined as general, not necessarily pointed, pseudo-triangulations with internal faces of size~3 or~4. We call a size-4 pseudo-triangle a \emph{dart}, with its \emph{tail} being the concave (also called reflex) vertex, its \emph{tip} being the vertex not adjacent to the tail, and its two \emph{wings} being the remaining two vertices. The segment between the tip and tail of a dart will be referred to as its \emph{spine}, though such a segment is necessarily not an edge and therefore missing in a dart. In a 4-PT, each interior pointed vertex is the tail of a dart, and 4-PTs with $k$ interior pointed vertices are 4-PTs with $k$ darts or \emph{$k$-dart 4-PTs}, which will be denoted as $k$-DPTs. See Figure~\ref{fig:dpt-example}.

For a size-$n$ point set~$P$ with a convex hull of size~$h \leq n$, the maximum number of interior pointed vertices is $n-h$, and therefore the maximum number of darts in a 4-PT is $n-h$ as well. In particular, 4-PPTs coincide with $k$-DPTs for $k=n-h$, since they are those 4-PTs in which every interior vertex is pointed. Thus, the aforementioned open problem in~\cite{aichholzer2014flips} asks about the connectivity of the flip graph of $k$-DPTs for the largest possible value of $k$, that is $k=n-h$. Our first goal is to look at the opposite end of the range and study the flip graph of $k$-DPTs for the smallest possible value of $k$, that is~$k=1$. This corresponds to 4-PTs with only one dart, i.e., only one interior pointed vertex.
We show that the resulting flip-graph of $1$-DPTs is not connected and we show how to compute its connected components.
Furthermore, for $k$-DPTs on a frequently-studied point configuration, the double chain~\cite{aichholzer2008number}, we analyze the structure of the flip graph on a more fine-grained level.

\begin{figure}
    \centering
    \includegraphics{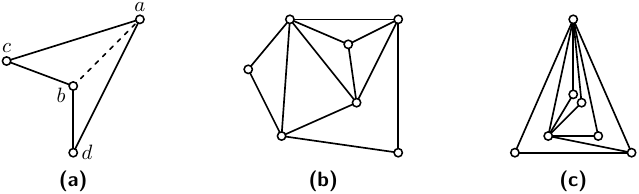}
    \caption{\textbf{\textsf{(a)}} A dart with tip~$a$, tail~$b$, wings~$c$ and~$d$, and a dashed spine. \textbf{\textsf{(b)}} A $1$-DPT on 7 points. \textbf{\textsf{(c)}} An $(n-h)$-DPT with $n=7$ and $h=3$; each vertex not on the convex hull is a dart~tail.}
    \label{fig:dpt-example}
\end{figure}


\section{Components of the Flip Graph for 1-DPTs}
To compute the number of components of the flip graph in the presence of a single dart, we first partition the class of $1$-DPTs on a point set~$P$ into separate classes~$\mathcal{G}_p$ where~$p\in P$. Each such class~$\mathcal{G}_p$ consists of all those $1$-DPTs that have the tail of the dart located at~$p$. We show that all $1$-DPTs in~$\mathcal{G}_p$ are in the same connected component of the flip graph. The proof consists of three steps: First, we show that for any dart~$d$ on~$P$ there exists a \emph{dart triangle}, defined as a triangle on the three corners of a dart containing no more points of~$P$ than the tail of that dart (see Figure~\ref{fig:dart-triangle}a), with the property that such a dart triangle shares the tip and tail with the original dart~$d$ (but may have other wings, see Figures~\ref{fig:dart-triangle}b-d).

\begin{lemma}\label{lem:dart-triangle}
    For any $1$-DPT of a point set~$P$ containing an arbitrary dart~$d$ there is a flip sequence to a $1$-DPT with a dart triangle that has the same tip and tail as~$d$.
\end{lemma}
\begin{proof}
    Consider an arbitrary $1$-DPT~$T$ with dart~$d$. If the wings of~$d$ are connected by an edge in~$T$, and the triangle formed by the wings and tail of~$d$ is empty, then we are done. Hence, assume that either the wings of~$d$ are not connected by an edge, or some vertices of $P$ are located in the triangle formed by the wings and tail of~$d$. We now consider a triangulation~$T'$ of~$P$ that contains all edges of~$T$, and additionally contains the spine of~$d$ as an edge. We refer to the edges of~$d$, as well as the spine edge, as the set~$D$. We will now define a flip sequence on $T'$ that considers only triangulations in which these edges of~$D$ are constrained, meaning that these edges cannot take part in edge flips. 
    
    First, consider the case in which the wings of~$d$ are not connected by an edge in~$T'$. There exists at least one triangulation~$T^-$ that contains all the edges in~$D$, as well as the edge between the wings of~$d$: Simply start with only these edges and arbitrarily add edges that do not cross any existing edges, until we reach a maximal planar graph on~$P$. Dyn et al.~\cite{DBLP:journals/cagd/DynGR93} proved that there is a sequence~$S$ of flips that transforms~$T'$ into~$T^-$ while keeping all edges in~$D$ in place. Since the spine edge is never moved, we can apply the flip sequence~$S$ to~$T$, to transform~$T$ into a pseudo-triangulation in which the wings of~$d$ are connected by an edge. If the triangle~$t$ defined by the tip and wings of~$d$ now contains only the tail of~$d$, then we have found a flip sequence that creates a dart triangle sharing its tip and tail with~$d$. Otherwise we are in the second case, thus what remains is to prove the second case.
    
    Second, consider the case where some subset of points of~$P$ are located in the triangle formed by the wings and tail of~$d$. We call this subset, together with the wings of~$d$, the set~$P'\subset P$. We can use the flip sequence of the first case to ensure that the wings of~$d$ are connected by an edge. See Figure~\ref{fig:dart-triangle}b for the construction. Assume without loss of generality that the dart~$d$ is oriented such that the spine, from tail to tip, points straight upwards. If this is not the case, simply rotate the complete point set until we achieve this configuration. Now extend the spine downwards, to split the points in~$P'$ in subsets~$L$ and $R$ to the left and right of the spine extension, respectively. We show that there is a flip sequence that ensures that $L$ and $R$ become singletons: They consist of only the wings of~$d$. In the base case, $L$ and $R$ are already singletons, so for the step case assume at least one of them has size at least two. When $L$ and/or~$R$ are of this size, we respectively find the points $l\in L$ and/or $r\in R$, that form an empty triangle with the tail and respective wing of~$d$.
    
    Consider the upper envelope of~$P'$. We choose $l\in L$ to be the rightmost upper envelop vertex in~$L$, and similarly $r\in R$ the leftmost upper envelope vertex in~$R$. We now describe the flip sequence to move the wings to~$l$ and~$r$. There is a pseudo-triangulation in which the upper envelope vertices are all connected to the tail of $d$, and from left to right, each consecutive pair of upper envelope vertices is also connected, resulting in a triangulation of the space between the upper envelope and~$d$ (see Figure~\ref{fig:dart-triangle}c). We can again use the result by Dyn et al.~\cite{DBLP:journals/cagd/DynGR93} to flip to this pseudo-triangulation by fixing the dart edges. Now consider each side of the spine extension separately. In this pseudo-triangulation, as long as the wings of~$d$ do not coincide with $l$ and $r$, we can flip the edge connecting the respective wing with the tail of~$d$, to bring that wing of $d$ closer to $l$ or $r$, one upper envelope vertex at a time. Hence, this flip sequence will create a dart triangle sharing its tip and tail with~$d$ (see Figure~\ref{fig:dart-triangle}d).

    To summarize, we can now perform the flip sequence of the former case to connect the wings of~$d$, followed by a flip sequence for the latter case (if necessary), to create a dart triangle sharing its tip and tail with~$d$.
\end{proof}

\begin{figure}[b]
    \centering
    \includegraphics{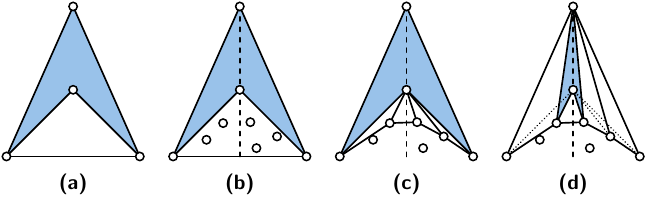}
    \caption{\textbf{\textsf{(a)}} A dart triangle. \textbf{\textsf{(b)}} A dart with an edge connecting the wings. The vertices in the bottom face are split by the extension of the (dashed) spine. \textbf{\textsf{(c)}} The specific triangulation between the dart and the upper envelop of the subset~$P'$ of points of~$P$ in the triangle formed by the wings and tail of~$d$, together with the wings of~$d$. \textbf{\textsf{(d)}} A number of flips linear in~$|P'|$ creates a dart triangle.}
    \label{fig:dart-triangle}
\end{figure}

Second, we show that for a $1$-DPT with the tail of the dart at point~$p_d\in P$, the tip can be flipped to any point~$p \in P\setminus \{p_d\}$ if this allocation of tip and tail permits a $1$-DPT of~$P$.

\begin{lemma}\label{lem:move-1-tip}
    If a point~$p_d\in P$ is the tail of the single dart in two $1$-DPTs~$\mathcal{T}_1$ and~$\mathcal{T}_2$ on point set~$P$, then there is a flip sequence between~$\mathcal{T}_1$ and~$\mathcal{T}_2$.
\end{lemma}
\begin{proof}
    We may assume that the dart in~$\mathcal{T}_1$ and in~$\mathcal{T}_2$ is inside a dart triangle: if this is not the case, we apply Lemma~\ref{lem:dart-triangle} and use the respective flip sequences to $1$-DPTs with dart triangles as a pre-/suffix for the flip sequence we construct during this proof, to flip~$\mathcal{T}_1$ to~$\mathcal{T}_2$.

    Since the dart in~$\mathcal{T}_1$ and~$\mathcal{T}_2$ is in a dart triangle, we can remove~$p_d$ from the triangulations, to get two (proper) triangulations~$\mathcal{T}_1'$ and~$\mathcal{T}_2'$. We mark the faces from which we removed~$p_d$ as a special triangular faces~$f_1^*$ and~$f_2^*$. As the flip graph for triangulations is connected~\cite{DBLP:journals/dm/Lawson72}, there is a flip sequence transforming~$\mathcal{T}_1'$ into~$\mathcal{T}_2'$. We can use the exact same flip sequence to transform~$\mathcal{T}_1$ into~$\mathcal{T}_2$, as long as the edges incident with a special face do not participate in a flip. What remains is to show that in an intermediate triangulation~$T_i'$, whenever an edge incident with the special face~$f^*$ participates in a flip, we can ensure the following: There is a flip sequence of $1$-DPTs that starts from pseudo-triangulation~$\mathcal{T}_i$, which is obtained by putting~$p_d$ back into the special face~$f^*$ of~$\mathcal{T}_i'$, before the flip involving~$f^*$ and ends with $p_d$ in a dart triangle corresponding to the (special) face after the flip involving~$f^*$. 
    
    Such a flip sequence can be seen in Figure~\ref{fig:dart-triangle-flips} and always exists for the following reason. Let the points $s,t,u,v$ be corner points of the quadrilateral affected by the edge flip(s), and let $su$ be the diagonal that is present initially. Since~$f^*$ is a dart triangle, we can rotate the dart so that the wings are incident with~$su$. The flip then involves two triangular faces in which we swap from one diagonal to the other. Since we assume general position, the tail~$p_d$ cannot be collinear with $t$ and $v$ and hence one of the edges between~$p_d$ and a wing point crosses~$tv$. This edge can flip to~$tv$ to finish the flip sequence.
\end{proof}

\begin{figure}[t]
    \centering
    \includegraphics{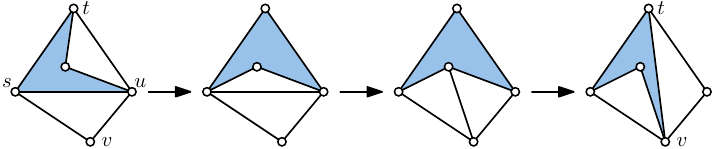}
    \caption{A flip sequence to flip an edge incident with the face containing the tail~$p_d$ of a dart.}
    \label{fig:dart-triangle-flips}
\end{figure}

The flip sequence of Lemma~\ref{lem:move-1-tip} implies that all $1$-DPTs in the class~$\mathcal{G}_p$ of $1$-DPTs that have the tail of the dart located at~$p$ reside in the same connected component of the flip graph. As the third and final step in finding the connected components of the flip graph for $1$-DPTs on~$P$, we still have to check whether the tail of the dart can move from one vertex~$p\in P$ to $q\in P$, i.e., whether the $1$-DPTs in $\mathcal{G}_p$ and $\mathcal{G}_q$ all reside in the same connected component. To do so, we consider, for each pair of vertices $p,q\in P$, each triple of points in $P$ distinct from $p$ and $q$ and check whether no other points are located in the faces of any of the small configurations in Figure~\ref{fig:move-tail}a-d. Observe that these five vertices must admit two overlapping darts with different tails. If this is the case, then by Lemma~\ref{lem:move-1-tip} we can flip to the 1-DPTs where the respective darts have their tips in the position prescribed in the small configuration, and perform the flip in the configuration to swap the tails. To complete the argument, we can use a finite case analysis to prove that we have to check only the configurations in Figure~\ref{fig:move-tail}a-d; no other ways to swap tails in 1-DPTs exist.

\begin{lemma}\label{lem:flip-tail-configs}
    There exist exactly four configurations of five points, that allow a dart to move its tail with an edge flip, as illustrated in Figure~\ref{fig:move-tail}a-d.
\end{lemma}
\begin{proof}
    First observe that for two darts with different tail vertices to exist in a size-5 point set, two points have to be located in the convex hull of the other points. Hence we consider only configurations with two points inside a triangular convex hull and all such configurations have the same order type. Therefore, w.l.o.g. assume the points are embedded as in Figure~\ref{fig:move-tail}.

    We first assume that the triangle shared by the darts is incident with the two points inside the convex hull and the topmost point on the convex hull. In this case, the right point inside the convex hull can be a dart by having the spine pointing towards the topmost point, or towards the other point inside the convex hull. Similarly the left point inside the convex hull has symmetric options for its spine. The four combinations of these options result the configurations in Figure~\ref{fig:move-tail}a-d, and allow the tails of the darts to swap between the points inside the convex hull with a single edge flip.

    Now consider any other triangle incident with the two points inside the convex hull, for example the blue triangle in Figure~\ref{fig:move-tail}e. One of the points inside the convex hull can create a dart with a spine pointing towards the other point in the convex hull, or towards the third point in the shared triangle (the left point in Figure~\ref{fig:move-tail}e). For the other point inside the convex hull neither of these options work out, since they require one of the wing vertices to be located in an area that does not contain a point (given the chosen point embedding/order).
\end{proof}
Lemmata~\ref{lem:move-1-tip} and~\ref{lem:flip-tail-configs} together result in the following theorem.
\begin{figure}[b]
    \centering
    \includegraphics{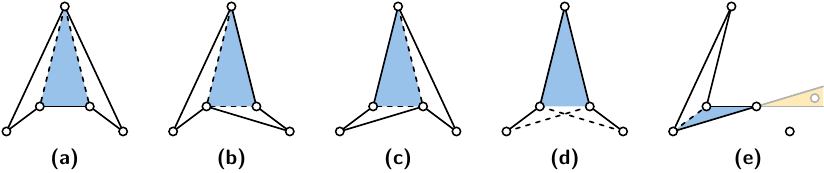}
    \caption{\textbf{\textsf{(a)-(d)}} All dart configurations of five points that allow us to move the tail of a dart using an edge flip. The darts share the blue triangle and each require one dashed edge. \textbf{\textsf{(e)}} Other triangles cannot use both middle vertices as tails, without a (non-existing) point in the yellow area.}
    \label{fig:move-tail}
\end{figure}
%
\begin{theorem}
For 1-DPTs on a set~$P$ of $n$ points with $h$ points on the convex hull, there are at most $n-h$ components in the flip graph. The exact components can be determined by checking every quintuple of points that have a triangular convex hull.
\end{theorem}

\section{Characterizing the Flip Graph for the Double Chain}

In this section we consider the double chain, a point set consisting of a convex $4$-gon being the hull of two concave chains of points next to opposite edges of the $4$-gon, $P_\asymp = P_1 \cup P_2$, such that these concave chains do not cross the diagonals of the $4$-gon, see Figure~\ref{fig:concave-chains}a. 
We can completely characterize the flip graph of $k$-DPTs on~$P_\asymp$, for any possible~$k$. 

A $k$-DPT on a point set~$P_\asymp$ admits two kinds of darts, \emph{aligned darts} for which the spine connects two adjacent vertices of one concave chain, and \emph{crossing darts} which have tip and tail in opposite concave chains. In this section we prove that the tail of a dart cannot swap between concave chains, and we say that a dart is \emph{designated} to the chain where the tail is located. Additionally, we show that we can use edge flips to flip any $k$-DPT of an instance~$P_\asymp$ to a \emph{canonical} $k$-DPT. In such a $k$-DPT, all darts are aligned, and flipped as far left as possible (see Figure~\ref{fig:concave-chains}b); for darts designated to~$P_1$ the wings on the opposite chain are at the leftmost point on~$P_2$, while for darts designated to~$P_2$ the analogous wings are at the rightmost tail on~$P_1$. Furthermore, all wings inside the convex hulls of their designated chains are located at the rightmost point of the chain. By analyzing the number of such canonical $k$-DPTs, we will analyze the number of connected components of the flip graph.

\begin{figure}
    \centering
    \includegraphics{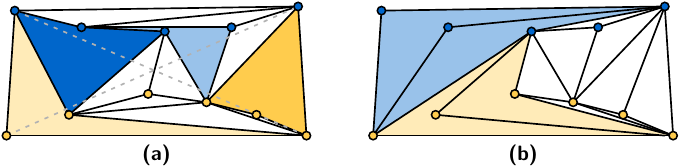}
    \caption{\textbf{\textsf{(a)}} A $k$-DPT for concave chains $P_1$ (blue) and $P_2$ (yellow), with $k=4$. The aligned and crossing darts are light and dark colored, respectively. Note that $P_1$ and $P_2$ do not cross the dashed-grey diagonals. \textbf{\textsf{(b)}} The canonical $k$-DPT of \textbf{\textsf{(a)}}; white faces are triangulated arbitrarily.}
    \label{fig:concave-chains}
\end{figure}

We first prove a few properties of aligned and crossing darts using geometric observations. 

\begin{lemma}\label{lem:aligned-dart}
    In a $k$-DPT of point set~$P_\asymp$, any aligned dart has one wing on the opposite concave chain, and for any choice of wings on the respective concave chains, a dart exists.
\end{lemma}
\begin{proof}
    For the first part of the statement, observe that four points of a single concave chain in~$P_\asymp$ do not admit a dart: Any quadruple of points in a concave chain forms a convex set.

    For the second part of the statement, we can use a consequence of the chains hulls not crossing the diagonals: If we connect the wings of an aligned dart, which must now lie on different concave chains, this edge does not cross the convex hulls of the individual concave chains. Now, for any choice of wing vertices, consider any two consecutive points on a concave chain (as tip and tail), disjoint from the wings. If we draw a line through the wings, then the two non-wing points lie on one side of the wings. Hence the line piece between the wing points would lie completely outside of the quadrilateral of the four considered points, in which the wings are not consecutive. This (arbitrarily chosen) quadrilateral must hence have one concave corner at one of the non-wing points, resulting in a aligned dart.
\end{proof}

\begin{lemma}\label{lem:crossing-dart}
    In a $k$-DPT of point set~$P_\asymp$, the wings and tail of any crossing dart are consecutive on one concave chain; for any choice of tip on the opposite chain, a dart exists.
\end{lemma}
\begin{proof}
    First observe that the tail must be between the wings on the concave chain, otherwise the angle between the edges towards the wings is not concave. Now, for the first part of the statement, assume for contradiction that the wings and tail of a crossing dart are not consecutive. Since the edges from the wings to the tip do not cross the convex hulls of the individual concave chains, the vertices between tail and wing will be encompassed by the edges of the dart.

    For the second part of the statement, we again use that an edge connecting vertices on opposite chains does not cross the convex hulls of the chains: When the wings and tail are consecutive, the wings can connect to any vertex on the opposite chain, and the resulting quadrilateral will not encompass any points, and will hence be a crossing dart
\end{proof}

\newpage
\begin{lemma}\label{lem:concave-darts}
    A $k$-DPT of point set~$P_\asymp$ admits only crossing and aligned darts.
\end{lemma}
\begin{proof}
    Assume for contradiction that $P_\asymp$ admits a dart with its spine not connecting two points on opposite concave chains, and not connecting two consecutive points on a single concave chain. Also, as any quadruple of points in a concave chain forms a convex set, one wing of this dart must lie on the opposite concave chain. Now consider the two edges from the tip and tail to the wing on the opposite chain. These edges will not cross the convex hull of the concave chain that contains the tip and tail.

    Since the spine does not connect consecutive vertices on the concave chain, there is at least one point~$p$ between the tip and tail on the concave chain. Therefore the point $p$ will be located inside the considered dart, because the edges to the wing on the opposite chain do not cross the convex hull of either chain, and hence encompass~$p$, leading to a contradiction.
\end{proof}

Again using quintuples of points, we prove that tails cannot swap between concave chains.

\begin{lemma}\label{lem:concave-no-tail-move}
    In any $k$-DPT of point set~$P_\asymp = P_1 \cup P_2$ the tail of a dart cannot swap between $P_1$ and $P_2$ through an edge flip.
\end{lemma}
\begin{proof}
    In the previous section we already argued that such a flip requires two points (the tail before and after the flip) to be inside the convex hull of three other points. Thus consider the convex hull of any 5 points in $P_\asymp$. If we take two or less points of one concave chain, then these points can never be inside the convex hull of these 5 points. Thus any quintuple of points in $P_\asymp$ can only contain points from one concave chain inside its convex hull. 
\end{proof}

Next we show how to flip any $k$-DPT on an instance~$P_\asymp$ to the canonical $k$-DPT for~$P_\asymp$.

\begin{lemma}\label{lem:concave-canonical}
    Any $k$-DPT of point set~$P_\asymp$ can be flipped to the canonical $k$-DPT.
\end{lemma}
\begin{proof}
    We first consider all darts with tails on~$P_1$ and flip them to be aligned darts on the left of the $k$-DPT. By Lemma~\ref{lem:concave-darts} these must be either crossing or aligned and we call these darts~$P_1$-darts. Afterwards we do the exact same for all darts with tails on~$P_2$ ($P_2$-darts).

    We consider the darts with wings on~$P_1$ from left-to-right, starting from dart with the leftmost tail. For each of these darts we move the tip or wing of the respectively crossing or aligned dart located on~$P_2$ to the leftmost vertex of~$P_2$. For this we consider the edges between~$P_1$ and~$P_2$. We do the following flips, until the currently considered dart becomes adjacent to the outer face, or to the previously considered dart. There are two cases:
    \begin{itemize}
        \item The face left of the currently considered dart has a single vertex on~$P_1$. In this case we can flip the edge between that face and the dart to the opposite diagonal (see Figure~\ref{fig:concave-moves-tips}). Note that this works regardless of whether the dart is crossing or aligned. It also always leads to a valid pseudo-triangulation, since Lemmata~\ref{lem:aligned-dart} and~\ref{lem:crossing-dart} tell us that any position of the dart vertex on the opposite chain is possible.

        \begin{figure}[b]
            \centering
            \includegraphics[page=1]{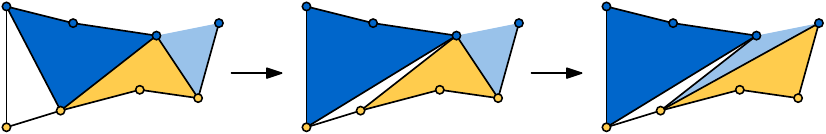}
            \caption{Moving the tip of a crossing dart, followed by moving the wing of an aligned dart.}
            \label{fig:concave-moves-tips}
        \end{figure}
        
        \item The face left of the currently considered dart has multiple vertices on~$P_1$ and the vertex on~$P_2$ is not the last one. Note that this must be a triangular (non-dart) face, as otherwise it is the previously considered dart, which should have its vertex on~$P_2$ at the leftmost point of~$P_2$. We now first flip the vertex on~$P_2$ of this triangular face to the leftmost point of~$P_2$, by recursively using the same process. Since the leftmost point on~$P_2$ was not yet reached, after the recursive procedure ends, the adjacent face must have two points on~$P_2$, and the former case applies again.
    \end{itemize}

    Once all $P_1$-darts have their vertex on $P_2$ moved all the way left, we turn them into aligned darts with their spines as far left along~$P_1$ as possible. We additionally want that the wing inside the convex hull of~$P_1$ is located at the rightmost point of~$P_1$. This works as follows (see Figure~\ref{fig:concave-moves-spines}).
    \begin{itemize}
        \item To turn a crossing dart into an aligned dart (or vice versa), simply flip the edge in the position of the aligned spine to the spine of the crossing dart (or vice versa).
        \item For a crossing dart that is as far left as possible, we have to make sure that the pseudo-triangulation inside the convex hull of~$P_1$ is such that the wing inside this convex hull is located at the rightmost point. Observe that inside this convex hull, we have triangular faces and one half of each aligned dart, which is also a triangle. We consider these half-darts as normal triangles, and flip the triangulation inside the convex hull of~$P_1$ to a triangulation where each half-dart triangle has its third, non-spine vertex located at the rightmost point of~$P_1$. By Lemma~\ref{lem:aligned-dart} we know that the wings of aligned $P_1$-darts may connect anywhere on~$P_1$, and thus the half-darts form proper darts with their other half at any point in during this flip sequence.
        \item Finally, to move the spine of a dart more towards the left, we turn a dart into a crossing dart (if it is not a crossing dart yet), and observe that the face left of the dart must be a triangular face, with two vertices on~$P_1$, otherwise we are done. We can flip the edge between this triangular face, and the dart, to the spine of the dart, to move the tail of the dart towards the left (see Figure~\ref{fig:concave-moves-crossing} for an example on a~$P_2$-dart). We do this until the leftmost position of the tail is reached, which necessarily is the leftmost point in~$P_1$ that is not on the convex hull of~$P_\asymp$, and not occupied by a previous considered dart. Here we can flip the crossing dart back to an aligned dart to finish processing this dart.
    \end{itemize}

    \begin{figure}
        \centering
        \includegraphics[page=2]{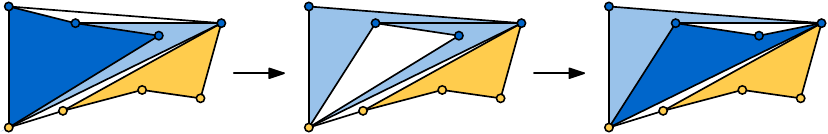}
        \caption{Flipping a crossing dart to be aligned, followed by a flip from an aligned to a crossing~dart.}
        \label{fig:concave-moves-spines}
    \end{figure}

    \begin{figure}[b]
        \centering
        \includegraphics[page=3]{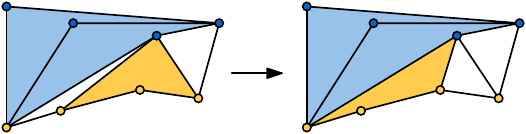}
        \caption{Moving a crossing dart designated to $P_2$ leftwards.}
        \label{fig:concave-moves-crossing}
    \end{figure}

    After the flip sequence described up to now, all~$P_1$-darts now coincide with the canonical $k$-DPT. To finish the flip sequence, find an analogous flip sequence for the $P_2$-darts, with one change: For each $P_2$-dart the vertex on $P_1$ does not flip to the leftmost vertex of~$P_1$ but to the rightmost point that is a tail of a $P_1$-dart (or equivalently, up to but not past the $P_1$-dart spines). An example of the resulting canonical $k$-DPT is shown in Figure~\ref{fig:concave-chains}b.
\end{proof}

Lemmata~\ref{lem:concave-no-tail-move} and~\ref{lem:concave-canonical} now allow us to prove the following theorem.

\begin{restatable}{theorem}{concavedistribute}\label{thm:concave-distribute}
    The number of connected components of the flip graph of $k$-DPTs on a point set~$P_\asymp = P_1 \cup P_2$ is equal to the number of ways to designate $k$ darts to $P_1$ or $P_2$.
\end{restatable}
\begin{proof}
    By Lemma~\ref{lem:concave-no-tail-move} we know that the tails of darts cannot swap between convex chains. Furthermore, Lemma~\ref{lem:concave-canonical} tells us that all such $k$-DPTs can be flipped to a canonical $k$-DPT. Now observe that the number of darts designated to $P_1$ and $P_2$ completely determines the canonical $k$-DPT. Thus all $k$-DPTs with the same designation of darts to the concave chains are part of the same connected component of the flip graph.
\end{proof}

%

If $P_1$ and $P_2$ have $a+2$ and $b+2$ points, respectively, with $a,b\ge0$, then there are $a+b$ points that can be a tail of a dart, and thus $k\leq a+b$. Distinguishing several cases depending on $a$, $b$, and $k$, we arrive at a unified formula for the number of components of the flip graph: Intuitively, we distribute $k$ indistinguishable balls over $2$  distinguishable (fixed-size) bins.


\begin{corollary}
    The number of connected components of the flip graph of $k$-DPTs on a point set~$P_\asymp = P_1 \cup P_2$, with $|P_1|=a+2$ and $|P_2|=b+2$, is equal to $\min\{a,b,k,a+b-k\}+1$ for $0\le k\le a+b$.
\end{corollary}

\section{Conclusion}
We studied the flip graph of pseudo-triangulations with faces of size 3 (triangles) and a bounded number~$k$ of size-4 faces (darts). For $k=1$ in any point configuration, and for any~$k$ in the double chain point configuration, we showed how to find the connected components of the flip graph. For general point configurations, we conjecture that a similar approach will work for slightly higher values of~$k$, such as $k\in\{2,3\}$. Our goal in studying these special configurations is to obtain new insights into the flip graph of pointed pseudo-triangulation with faces of size 3 or 4. However, our current findings do not seem to allow us to reach much further than low values of~$k$, as the required number of cases for higher~$k$ becomes infeasible.



\bibliography{references}

\end{document}